\pdfoutput=1
\documentclass{article} 
\usepackage{times}
\usepackage{geometry}

\usepackage[utf8]{inputenc} 
\usepackage[T1]{fontenc}    
\usepackage[colorlinks,citecolor=blue]{hyperref}       
\usepackage{url}            
\usepackage{booktabs}       
\usepackage{amsfonts}       
\usepackage{nicefrac}       
\usepackage{microtype}      

\usepackage{url}
\usepackage{authblk}

\usepackage{graphicx}
\usepackage{multirow}
\usepackage{floatrow}
\newfloatcommand{capbtabbox}{table}[][\FBwidth]

\PassOptionsToPackage{numbers}{natbib}
\usepackage{natbib}

\usepackage{enumerate}
\usepackage{amsthm,amsmath,amssymb}
\usepackage{amsfonts,dsfont}
\usepackage{nicefrac}
\usepackage{microtype}
\usepackage{mathtools}

\usepackage{caption}
\usepackage{subcaption}
\usepackage{enumitem}
\usepackage{algorithm}
\usepackage{algorithmic}
\usepackage[normalem]{ulem}
\usepackage{amssymb}
\usepackage{multicol}
\usepackage{adjustbox}
\usepackage{multirow}
\usepackage{color}
\usepackage{xspace}
\usepackage{mkolar_definitions}

\newcounter{qcounter}

\newtheorem*{theorem*}{Theorem}

\newcommand{\papscore}{paper score\xspace}
\newcommand{\papscores}{paper scores\xspace}
\newcommand{\iralg}{\textsc{FairIr}\xspace}
\newcommand{\flowalg}{\textsc{FairFlow}\xspace}

\newcommand{\covcstr}{\ensuremath{C}\xspace}

\newcommand{\loadub}{\ensuremath{U}\xspace}
\newcommand{\loadlb}{\ensuremath{L}\xspace}

\newcommand{\tpms}{\textsc{TPMS RAP}\xspace}
\newcommand{\prfa}{\textsc{PR4A}\xspace}
\newcommand{\midl}{\textsc{MIDL}\xspace}
\newcommand{\cvpra}{\textsc{CVPR}\xspace}
\newcommand{\cvprb}{\textsc{CVPR2018}\xspace}

\begin{document}
\title{Paper Matching with Local Fairness Constraints}

\author{Ari Kobren, Barna Saha, Andrew McCallum}
\affil{College of Information and Computer Sciences \\ University of Massachusetts Amherst}
\affil{\{akobren,barna,mccallum\}@cs.umass.edu}





\maketitle

\begin{abstract}
Automatically matching reviewers to papers is a crucial step of the peer review process for venues receiving thousands of submissions. Unfortunately,  common paper matching algorithms often construct matchings suffering from two critical problems: (1) the group of reviewers assigned to a paper do not collectively possess sufficient expertise, and (2) reviewer workloads are highly skewed. In this paper, we propose a novel \emph{local fairness formulation} of
paper matching that directly addresses both of these issues. Since optimizing
our formulation is not always tractable, we introduce two new algorithms, \iralg\ and \flowalg, for computing fair matchings that approximately optimize the new formulation. \iralg solves a relaxation of the local fairness formulation and then employs a rounding technique to construct a valid matching that provably maximizes the objective and only compromises on fairness with respect to reviewer loads and papers by a small constant. In contrast, \flowalg is not provably guaranteed to produce fair matchings, however it can be 2x as efficient as \iralg and an order of magnitude faster than matching algorithms that directly optimize for fairness.
Empirically, we demonstrate that both \iralg and \flowalg improve fairness over standard matching algorithms on real conference data. Moreover, in comparison to state-of-the-art matching algorithms that optimize for fairness only, \iralg achieves higher objective scores, \flowalg achieves competitive fairness, and both are capable of more evenly allocating reviewers.
\end{abstract}

\section{Introduction}
In 2014, the program chairs (PCs) of the Neural Information Processing
Systems (NeurIPS) conference conducted an experiment that allowed them to
measure the inherent randomness in the conference's peer review
procedure. In their experiment, 10\% of the
submitted papers were assigned to \emph{two} disjoint sets of
reviewers instead of one. For the papers in this experimental set, the
PCs found that the two groups assigned to review the same paper
disagreed about whether to accept or reject the paper 25.9\% of the
time. Accordingly, if all  2014 NeurIPS submissions were reviewed again by a new set of reviewers, about 57\% of the originally accepted papers would be rejected
\cite{price2014nips}.

The NIPS experiment is only one of many studies highlighting the poor reliability of the peer reviewing process. For example,
another study finds that the rate of agreement between reviewers for a
clinical neuroscience journal is not significantly different from
chance~\cite{rothwell2000reproducibility}. This is particularly troublesome given that decisions regarding patient care, expensive scientific exploration, researcher hiring, funding, tenure, etc. are all based, in part, on
published scientific work and thus on the peer reviewing process.

Unsurprisingly, previous work shows that experts are able to produce
higher quality reviews of submitted publications than non-experts.
Experts are often able to develop more ``discerning'' opinions about
the proposals under review \cite{johnson1982multimethod,
  camerer199710} and some
researchers in cognitive science and artificial intelligence claim
that experts can make more accurate decisions than non-experts about
uncertain information~\cite{johnson1988expertise}. Clearly 
 peer review outcomes are likely to be of higher quality if each paper
were reviewed exclusively by experts in the paper's topical areas. Unfortunately, since experts are relatively scarce, this is often impossible. Especially for many computer science venues, which are faced with increasingly large volumes of submissions, assigning only experts to each submission is impossible given
typical reviewer load restrictions. Further exacerbating the problem, conference decision processes are dictated by a strict timeline. This necessitates significant automation in matching reviewers to submitted papers, highly limiting the extent to which humans can significantly intervene.

Automated systems often cast the paper matching problem as a global maximization of reviewer-paper \emph{affinity}. In particular, each reviewer-paper pair has an associated affinity score, which is typically computed from a variety of factors, such as: expertise, area chair recommendations, reviewer bids, subject area matches, etc. The optimal matching is one that maximizes the sum of affinities of assigned reviewer-paper pairs, subject to \emph{load} and \emph{coverage} constraints, which bound the number of papers to which a reviewer can be assigned and dictate the number of reviews each paper must receive, respectively~\cite{charlin2012framework, TaylorTR08}. While optimizing the global objective has merit, a major disadvantage of the approach is that it can lead to matchings that contain papers assigned to a set reviewers who lack expertise in that paper's topical areas~\cite{garg2010assigning, stelmakh2018peerreview4all}. This is because in constructing a matching that maximizes the global objective, allocating more experts to one paper at the expense of another may improve the objective score. In order to be fair, it is important to ensure that each paper is assigned to a group of reviewers who instead possess a minimum acceptable level of expertise.

Recent work has attempted to overcome these problems by either (a) introducing strict requirements on the minimum affinity of valid paper-reviewer matches, or (b) optimizing the sum of affinities of the one paper that is worst-off~\cite{garg2010assigning, stelmakh2018peerreview4all}. However, restricting the minimum allowable affinity often renders the problem infeasible as there may not exist any matching that provides sufficient coverage to all papers subject to the threshold. Previously proposed algorithms that maximize the sum affinities for the worst-off paper do result in matchings that are more fair, but they also suffer from two disadvantages: (1) they do not simultaneously optimize for the overall best assignment (measured by sum total affinity), and (2) they are agnostic to lower limits on reviewer loads (which are common in practice) and thus may produce matchings in which reviewers are assigned to dramatically different numbers of papers.

To address these issues, we introduce the \emph{local fairness formulation} of the paper matching problem.  Our novel formulation is cast as an integer linear program that (1) optimizes the global objective, (2) includes both upper \emph{and lower} bound constraints that serve to balance the reviewing load among reviewers, and (3) includes \emph{local fairness constraints}, which ensure that each paper is assigned to a set of reviewers that collectively possess sufficient expertise. 

The local fairness formulation is NP-Hard. To address this hardness, we present \iralg, the \textbf{FAIR} matching via \textbf{I}terative \textbf{R}elaxtion algorithm that jointly optimizes the global objective, obeys local fairness constraints, and satisfies lower (and upper) bounds on reviewer loads to ensure more balanced allocation. \iralg works by solving a relaxation of the local fairness formulation and rounding the corresponding fractional solution using a specially designed procedure. Theoretically, we prove that matchings constructed by \iralg may only violate the local fairness and load constraints by a small margin while maximizing the global objective. In experiments with data from real conferences, we show that, despite theoretical possibility of constraint violations, \iralg never violates reviewer load constraints. The experiments also reveal that matchings computed by \iralg exhibit higher objective scores, more balanced allocations of reviewers and competitive treatment of the most disadvantaged paper when compared to state-of-the-art approaches that optimize for fairness.

In real-conference settings, a program chair may desire to construct and explore many alternative matchings with various inputs, which demands an efficient fair matching algorithm.  Toward this end, we present \flowalg, a min-cost-flow-based heuristic for constructing fair matchings that is faster than \iralg by more than 2x. While matchings constructed by \flowalg are not guaranteed to adhere to a specific degree of fairness (like \iralg or previous work), in experiments, \flowalg often constructs matchings exhibiting fairness and objective scores close to that of \iralg in a fraction of the time. Unlike \iralg and matching algorithms that rely on linear programming, \flowalg operates by first maximizing the global objective and then refining the corresponding solution through a series of min-cost-flow problems in which reviewers are reassigned from the most advantaged papers to the most disadvantaged papers.

This paper is organized as follows. Section~\ref{sec:matching} presents the standard paper matching formulation that optimizes the global objective. Section~\ref{sec:local} covers our main contribution by providing the local fairness formulation of paper matching and describes \iralg and its formal guarantees. Section~\ref{sec:flow} presents the more efficient \flowalg heuristic. In Section~\ref{sec:exp}, we experimentally show the effectiveness of our approach over other approaches on several datasets coming from real conferences. 
\section{Reviewer Assignment Problem}
\label{sec:matching}
Popular academic conferences typically receive thousands of paper submissions.
Immediately after the submission period closes, papers are automatically
matched to a similarly sized pool of reviewers.  
A \emph{matching} of reviewers to papers is constructed using real-valued reviewer-paper \emph{affinities}. 
The affinity between a reviewer and a paper may be computed from a variety of factors, such as: expertise, bids, area chair recommendations, subject area matches, etc. 
Previous work has explored approaches for modeling reviewer-paper affinity via latent semantic indexing, collaborative filtering or information retrieval techniques~\cite{dumais1992automating, charlin2012framework, conry2009recommender}. 
We do not develop affinity models in this work. 
Instead, we focus on algorithms for matching papers to reviewers given the affinity scores. 
In the literature, this matching problem is known by many names; we choose \emph{the reviewer assignment problem} (RAP)~\cite{kou2015weighted, stelmakh2018peerreview4all}.

The RAP is often accompanied by a two types of constraints:
\emph{load constraints} and \emph{coverage constraints}
\cite{garg2010assigning}. 
A load constraint bounds the number of papers assigned to a reviewer; a coverage constraint defines the number of reviews a paper must receive.
Typically, all papers must be reviewed the same number of times. Reviewers do not always have equal loads, although a highly uneven load is inherently unfair and may lead to reviewers declining to review or not submitting reviews on time.

Formally, let $R=\{r_i\}^{N}_{i=1}$ be the set of reviewers, $P=\{p_j\}^{M}_{j=1}$ be the set of papers and $\mathbf{A} \in \mathbb{R}^{|R|\times |P|}$ be a matrix of reviewer-paper affinities.
The RAP can be written as the following integer program:

\begin{align*}
  \max\quad& \sum_{i=1}^{|R|}\sum_{j=1}^{|P|}
  x_{ij}\mathbf{A}_{ij}&\\
  \text{subject }\text{to }\quad & \sum_{j=1}^{|P|}x_{ij}
  \le \loadub_i ,\ &\forall i=1,2,...,|R|\\
  &\sum_{i=1}^{|\Rcal|}x_{ij}
  = \covcstr_j ,\ &\forall j=1,2,...,|P|\\
  &x_{ij} \in \{0,1\},\ & \forall i,j \label{eq:integrality}.
\end{align*}
Here, $\{\loadub_i\}_{i=1}^{|R|}$ is the set of upper bounds on reviewer loads, and $\{\covcstr_j\}_{j=1}^{|P|}$ represents the coverage constraints. 
The matching of reviewers to papers is encoded in the
variables $x_{ij}$, where $x_{ij} = 1$ indicates that reviewer $r_i$
has been assigned to paper $p_j$. 
In this formulation, the objective is to maximize the sum of affinities of reviewer-paper assignments (subject to the constraints); it can be solved optimally in polynomial time with standard tools~\cite{TaylorTR08}.

In practice, lower bounds on reviewer loads are often invoked in order to spread the reviewing load more equally across reviewers.
The formulation above can be augmented to include the lower bounds by adding the following constraints: 
\begin{align*}
    \sum_{j=1}^{|P|}x_{ij} \ge \loadlb_i ,\ &\forall i=1,2,...,|R|,
\end{align*}
where $\{\loadlb_i\}_{i=1}^{|R|}$ is the set of lower bounds on reviewer loads.
The resulting problem is still efficiently solvable.
Note that the formulation above, with and without lower bounds, is currently employed by various conferences and conference management software, for example: TPMS, OpenReview, CMT and HotCRP~\cite{charlin2013toronto, stelmakh2018peerreview4all}.
We will henceforth refer to the above two formulations as the \tpms, where the inclusion of lower bounds will be clear from context.
\section{Fair Paper Matching}
\label{sec:local}
It is well-known that optimizing the \tpms can result in unfair matchings~\cite{garg2010assigning, stelmakh2018peerreview4all}.
To see why, consider the example RAP in Figure~\ref{fig:assign-ex}, in which there are 4 papers and 4 reviewers, and define the \emph{\papscore} for paper $p$ to be the sum of affinities of reviewers assigned to paper $p$.
In the example, each paper must be assigned 2 reviewers and each reviewer may only be assigned up to 2 papers.
Even though the matchings in Figures~\ref{fig:unfair-assign} and \ref{fig:fair-assign} obtain equivalent
objective scores under the \tpms, the matching in Figure
\ref{fig:unfair-assign} causes papers $P3$ and $P4$ to have much lower \papscores than papers $P1$ and $P2$.
In practice, this may indicate that $P3$ and $P4$ have been assigned to a collection of reviewers, none of whom are well-suited to provide an expert evaluation.
The assignment in Figure \ref{fig:fair-assign} is clearly more equitable with respect to the papers (and reviewers), but the 
\tpms  does not prefer this matching since it seeks to globally optimize affinity.

\begin{figure}[t!]
\centering
    \begin{subfigure}[b]{0.23\textwidth}
        \centering
        \includegraphics[width=0.85\textwidth]{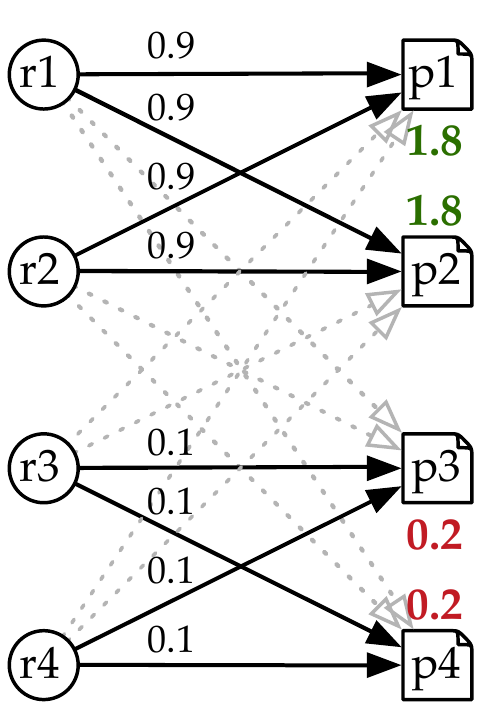}
        \caption{Unfair Matching.}
        \label{fig:unfair-assign}
    \end{subfigure}%
    \begin{subfigure}[b]{0.23\textwidth}
        \centering
        \includegraphics[width=0.85\textwidth]{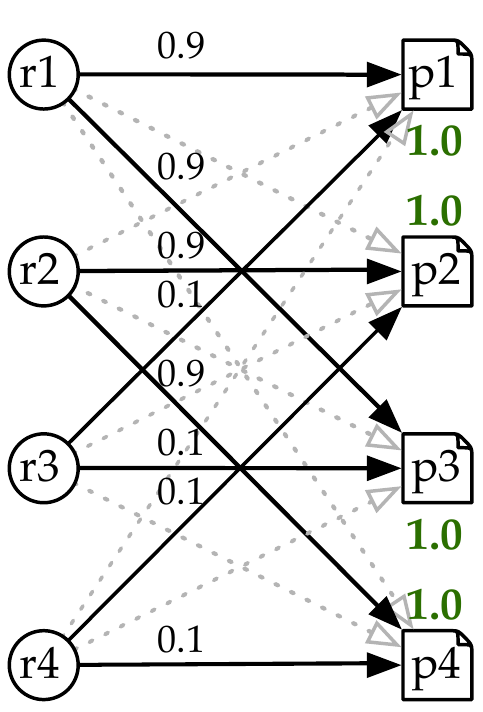}
        \caption{Fair Matching.}
        \label{fig:fair-assign}
    \end{subfigure}\\
    \caption{\textbf{Fair \& Unfair Matchings}. 4 papers and 4 reviewers with $\loadub=\loadlb=\covcstr=2$. $r1$ and $r2$ have affinity $0.9$ with all papers; $r3$ and $r4$ have affinity $0.1$ with all papers. Solid lines indicate assignments, the bold number adjacent to each paper corresponds to its \papscore. Matchings in both Figures \ref{fig:unfair-assign} and
      \ref{fig:fair-assign} achieve equivalent objective scores, but in Figure \ref{fig:unfair-assign} $p3$ and $p4$ are only assigned to reviewers with low affinity.}
\label{fig:assign-ex}
\end{figure}

\subsection{Local Fairness Constraints}
\label{sec:local fairness}
We propose to prohibit such undesirable matchings by augmenting the \tpms with \emph{local fairness constraints}.
That is, we constrain the \papscore at each paper to be no less than $T$~\cite{vazirani2013approximation}.
Formally, 
\begin{align*}
  \sum_{i=1}^{|R|}x_{ij}A_{ij} \ge T ,\ &\forall j=1,2,...,|P|.\\
\end{align*}
We refer to the resulting RAP formulation as the \emph{local fairness formulation}. 
While adding local fairness constraints is simple, this formulation is NP-Hard since it generalizes the max-min fair allocation problem \cite{vazirani2013approximation}.
To avoid the hardness of the local fairness formulation, one might instead be tempted to constrain the minimum affinity of valid assignments of reviewers to papers.
However, doing so often results in infeasible assignment problems~\cite{wang2008survey}.

\subsection{\iralg}
We present \iralg, an approximation algorithm for solving the local fairness formulation.
The algorithm is capable of accepting both lower and upper bound constraints on reviewer loads (as well as coverage constraints). 
By nature of being approximate, \iralg is guaranteed to return a matching in which any local fairness constraint may be violated by at most $A_{max} = \max_{r \in R, p \in P}A_{rp}$---the highest reviewer-paper affinity, and any reviewer load constraint (upper or lower bound) is violated by at most 1. Moreover, it achieves an $1$-approximation (no violation) in the global objective.
We call attention to the fact that our guarantees hold even though \iralg is able to accommodate constraints on reviewer lower bounds while optimizing a global objective, unlike most state-of-the-art paper matching algorithms with theoretical guarantees~\cite{garg2010assigning,stelmakh2018peerreview4all}.
Note that in practice lower bounds are often an input to the RAP in order to spread the reviewing load more equally across reviewers.

\begin{algorithm}[t]
  \caption{\iralg$(\Pcal)$}
  \label{alg:iralg}
  \begin{algorithmic}
    \STATE $\Pcal' \leftarrow \texttt{relax}(\Pcal)$
    \STATE $X \gets \{x_{ij} | i \in [|R|], j \in [|P|]\}$
    \WHILE{$X$ \texttt{is not empty}}
        \STATE $s \gets $\texttt{maximize}$(\Pcal')$
        \FOR{$x_{ij} \in X$}
        \IF{$x_{ij} == 0$}
            \STATE \texttt{fix} $x_{ij} = 0; \quad X \gets X \backslash x_{ij}$
        \ENDIF
        \IF{$x_{ij} == 1$}
            \STATE \texttt{fix} $x_{ij} = 1; \quad X \gets X \backslash x_{ij}$
        \ENDIF
        \ENDFOR
    \IF{$p \in P$ \texttt{has at most $3$ fractional assignments}}
        \STATE $\texttt{drop-fairness-constraint}(p)$
    \ENDIF
    \IF{\texttt{no fairness constraints dropped}}
        \IF{$r \in R$ \texttt{has at most $2$ fractional assignments}}
            \STATE $\texttt{drop-loads}(r)$
        \ENDIF
    \ENDIF
    \ENDWHILE
  \end{algorithmic}
\end{algorithm}

Our algorithm proceeds in rounds.
In each round, \iralg relaxes the integrality constraints of the local fairness formulation (i.e., each $x_{ij}$ can take
any value in the range $[0,1]$) and solves the resulting linear program. 
Any $x_{ij}$ with an integral assignment (i.e., either $0$ or $1$) is constrained to retain that value in subsequent rounds. 
Among the $x_{ij}$s with non-integral values, \iralg looks for a paper such that at most 3 reviewers have been fractionally assigned to it (the paper may have any number of integrally assigned reviewers). 
If such a paper is found, \iralg drops the corresponding local fairness constraint. 
If no such paper is found, \iralg finds a reviewer with at most 2 papers fractionally assigned to it and drops
the corresponding load constraints. 
The next round proceeds with the modified program. 
As soon as a matching is found that contains only integral assignments, that matching is returned. 
Algorithm \ref{alg:iralg} contains pseudocode for \iralg.

\begin{theorem}
  Given a feasible instance of the local fairness formulation
  $\Pcal = <R, P, \loadlb, \loadub, \covcstr,A,T>$, \iralg always terminates and returns
  an integer solution in which each local fairness constraint may be
  violated by at most $A_{max}$, each load constraint may be
  violated by at most 1 and the global objective is maximized.
  \label{thm:imma}
\end{theorem}
\noindent The proof of Theorem \ref{thm:imma} is found in the appendix. 

Theorem \ref{thm:imma} requires that the instance of the local fairness formulation be feasible. A RAP instance may be \emph{infeasible} if $T$ is too large, or if $\sum_{i=0}^{|R|}\loadub_i < \sum_{j=0}^{|P|}\covcstr_j$.
Checking the second condition is trivial.
To check if $T$ is too large, simply check if the corresponding relaxed local fairness formulation is infeasible. By Algorithm \ref{alg:iralg}, if the relaxed program is feasible, then \iralg must return an integer solution for that instance. Formally,

\begin{fact}
\label{fact:frac}
    If an instance of the local fairness formulation, $\Pcal$, is feasible after the integrality constraints on $x_{ij}$s have been removed, then Algorithm \ref{alg:iralg} returns an integral (possibly approximate) solution.
\end{fact}
\noindent Thus, by Fact \ref{fact:frac}, testing whether or not \iralg will return an integer solution for an instance of the local fairness formulation requires solving the relaxed program. In practice, a binary search over the feasible range of $T$ is performed and the highest $T$ yielding a feasible program is selected. Such a binary search requires solving the relaxed formulation several times and can add to the computational complexity. Overall, the running time of the algorithm is dominated by the number of times the linear program solver is invoked. Note that during each iteration of \iralg, many constraints may be dropped, which helps to improve scalability without sacrificing the theoretical guarantees. Also, note that by dropping constraints during each iteration the objective score can only increase.
\section{Faster Flow-based Matching}
\label{sec:flow}
\begin{figure*}
\begin{subfigure}[h]{0.4\columnwidth}
  \centerline{\includegraphics[width=0.75\columnwidth]{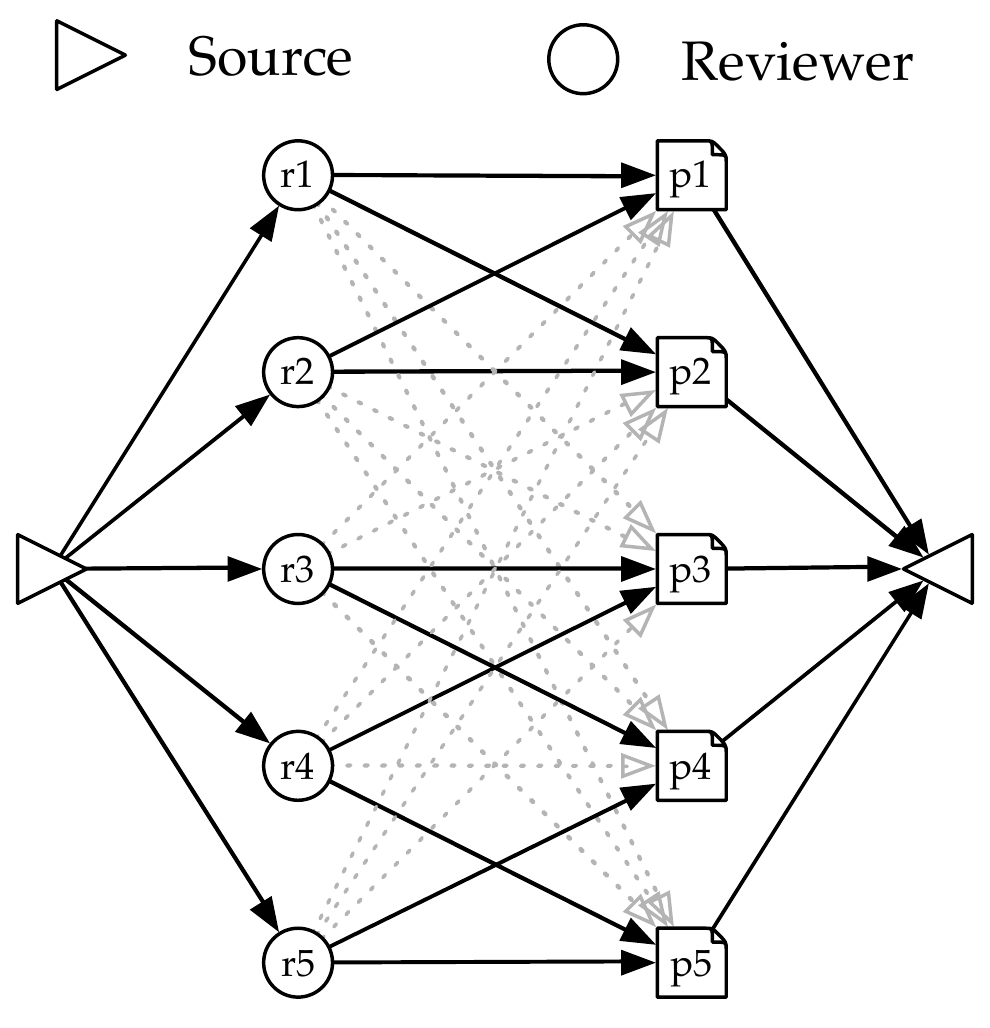}}
  \caption{RAP as flow.}
  \label{fig:rapflow}
\end{subfigure}
\begin{subfigure}[h]{0.4\columnwidth}
  \centerline{\includegraphics[width=0.75\columnwidth]{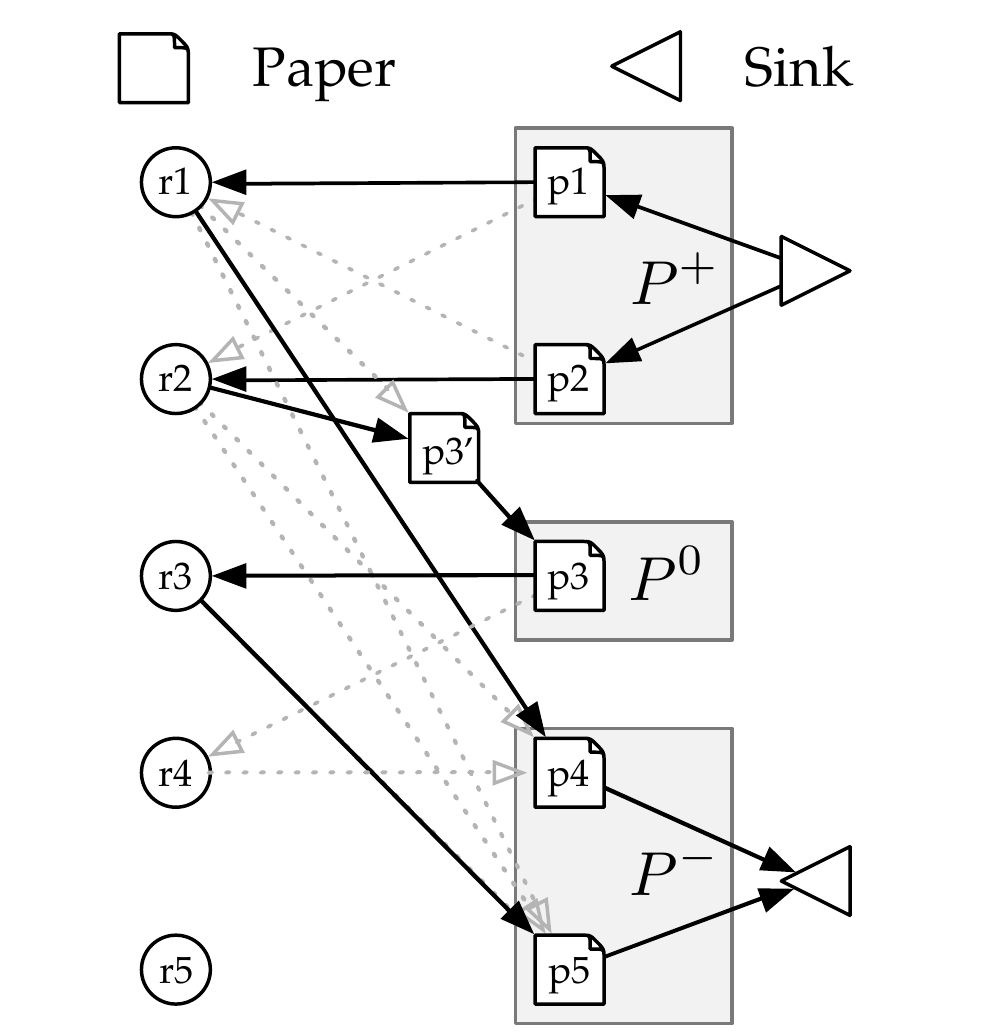}}
  \caption{Refinement Network.}
  \label{fig:refinenet}
\end{subfigure}
\caption{\textbf{\flowalg}. Darker bold arrows indicate utilized edges, lighter dotted arrows indicate edges that are not utilized. A darker edge originating at a reviewer $r$ and ending at a paper $p$ corresponds to the assignment of $r$ to $p$. A darker edge originating at a paper $p$ and ending at a reviewer $r$ corresponds to the \emph{"unassignment"} of reviewer $r$ from $p$. In the first step of \flowalg (Figure \ref{fig:rapflow}), MCF is solved in $\Gcal$. Papers are then grouped into $P^+$, $P^0$, and $P^-$ and the refinement network $\Gcal'$ is constructed. Edges are constructed to route flow from $P^+$ to $P^-$. Notice dummy node, $p3'$, which serves to limit the flow to $p3$.}
\label{fig:flowfigs}
\end{figure*}

For real conferences, paper matching is an interactive process.
A PC may construct one matching, and upon inspection, decide to tune the affinity matrix, $A$ and compute a new matching.
Alternatively, a PC may browse a matching and decide that certain reviewers should not be assigned to certain papers, or, that certain reviewers \emph{must} review certain papers.
After imposing the additional constraints, ideally, a new matching could be constructed efficiently.

\iralg is founded on solving a sequence of linear programs, and thus may not be efficient enough to support this kind of interactive paper matching when the number of papers and reviewers is large.
Other similar algorithms, which consider local constraints, also may not be efficient enough because they too rely on linear programming solvers \cite{garg2010assigning,stelmakh2018peerreview4all}.
Therefore, we introduce a min-cost flow-based heuristic for solving the local fairness formulation that is significantly faster than other state-of-the-art approaches.
While our flow-based approach does not enjoy the same performance guarantees of \iralg, empirically, we observe that it constructs high quality matches on real data (Section \ref{sec:exp}).

\subsection{Paper Matching as Min-cost Flow}
\label{subsec:mcf}
We begin by describing how to solve the \tpms using algorithms for \emph{min-cost flow} (MCF).
Our first focus is on RAP instances without constraints on reviewer load lower bounds.
Then we describe briefly how load lower bounds can be incorporated.

Construct the following graph, $\Gcal$, in which each edge has
both an integer cost and capacity:
\begin{enumerate}
\item create a source node $s$ with supply equal sum over papers of the corresponding coverage constraint: $\sum_{j=1}^{|P|} \covcstr_j$;
\item create a node for each reviewer $r_i \in R$ and a directed edge between the $s$ and each reviewer node with capacity $\loadub_i$ and cost $0$;
\item create a node for each paper $p_j
  \in P$ and create a directed edge from each reviewer, $r_i$, to each paper with cost $-A_{ij} \cdot W$, where $W$ is a large positive number to ensure that the cost of each edge is integer. Each such edge has capacity $1$;
\item construct a sink node $t$ with demand equal to the supply at $s$; create a directed edge from each paper $p_j \in P$ to $t$ with capacity $\covcstr_j$ and cost $0$.
\end{enumerate}
Then, solve MCF for $\Gcal$, i.e., find the set of edges in $\Gcal$ used in sending a maximal amount of flow from $s$ to $t$ such that, for each edge $e$, no more flow is sent across $e$ than $e$'s capacity, and such that the sum total cost of all utilized edges is minimal.
Note that algorithms like Ford-Fulkerson can be used to solve MCF and many efficient implementations are publicly available.
It can be shown that the optimal flow plan on this graph corresponds to the optimal solution for the \tpms.
In particular, each edge between a reviewer and paper utilized in the optimal flow plan corresponds to an assignment of a reviewer to a paper.
See Figure \ref{fig:rapflow} for a visual depiction of the $\Gcal$.
 
\subsection{Locally Fair Flows}
\label{subsec:local-flow}
We introduce a MCF-based heuristic, \flowalg, for approximately solving the local fairness formulation via a sequence of MCF problems.
Our algorithm is inspired by the combinatorial approach for (approximately) solving the scheduling problem on
parallel machines~\cite{gairing2007faster}.
\flowalg is comprised of three phases that are repeated until convergence.
In the first phase, a valid assignment is computed and the papers are partitioned into groups; in the second phase, specific assignments are dropped; in the third phase, the assignment computed in the first phase is refined to promote fairness.

In more detail, in phase 1 of \flowalg, $\Gcal$ is constructed using the 4 steps above (Section \ref{subsec:mcf}) and an assignment is constructed using MCF.
Afterwards, the papers are partitioned into three groups as follows:

\[
G(p_j) =
\begin{cases}
  P^{+} & \sum_{r_i \in R}x_{ij}A_{ij} \ge T \\
  P^{0} & T > \sum_{r_i \in R}x_{ij}A_{ij} \ge T - A_{max} \\
  P^{-} & \mathrm{otherwise}.
\end{cases}
\]
In words, the first group contains all papers whose \papscore is
greater than or equal to $T$; the second group contains all papers not in $P^{+}$ but whose \papscore is greater than $T$ \emph{minus} the maximum score; the third group contains all other papers.

In the second phase, for each paper $p \in P^{-}$ the reviewer assigned to that paper in phase 1 with the lowest affinity is \emph{unassigned} from $p$.

In the third phase, a \emph{refinement network}, $\Gcal'$, is constructed.
At a high-level, the refinement network routes flow from the papers in $P^{+}$ back through their reviewers and eventually to the papers in $P^{-}$ with the goal of reducing the number of papers with \papscores less than $T-A_{max}$.
The network is constructed as follows:
\begin{enumerate}
\item create a source node, $s$, with supply equal to the minimum among the number of papers in $P^{+}$ and $P^{-}$;
\item create a node for each $p \in P$; for each $p \in P^{+}$, create an edge from $s$ to $p$ with capacity $1$ and cost $0$;
\item create a node for each reviewer $r \in R$;
\item for each paper $p \in P^{+}$, create an edge with capacity $1$ and cost $0$ from $p$ to each reviewer assigned to $p$;
\item for each paper $p \in P^0$, create a dummy node, $p'$ and construct an edge from $p'$ to $p$ with capacity $1$ and cost $0$.
\item for each reviewer, $r$ assigned to a paper in $P^{+}$, create an edge with capacity $1$ and cost $0$ to each dummy paper, $p'$, if $r$ was not assigned to the paper to which $p'$ is connected;
\item for each paper $p \in P^{0}$ with dummy node $p'$, let $S_p$ be the current paper score at $p$, let $R(p')$ be the set of reviewers with edges ending at $p'$ and let $R(p)$ be the set of reviewers currently assigned to $p$. Let $A_{min}$ be the minimum affinity among the reviewers in $R(p')$ with respect to $p$. For each $r \in R(p)$ construct an edge with capacity $1$ and cost $0$ from $p$ to each $r$ if $T - A_{max} \le  S_p + A_{min} - A_{rp}$;
\item for each reviewer, $r$, construct an edge with capacity $1$ to each paper in $p \in P^{-}$ if $r$ is not currently assigned to that paper. If assigning $r$ to $p$ would cause $p$'s group to change to $P^{0}$, the cost of the edge is $-A_{rp} \cdot Z$, where $Z >> W$; otherwise, the cost is $-A_{rp} \cdot W$ (again, $Z$ is a large constant that ensures that edge costs are integral);
\item create a sink node $t$ with demand equal to the supply at $s$; for each paper $p \in P^{-}$ construct an edge from $p$ to $t$ with capacity $1$ and cost $0$.
\end{enumerate}
A visual illustration of the refinement network appears in Figure   ~\ref{fig:refinenet}.

After the network is constructed, MCF in $\Gcal'$ is solved.
The MCF in the refinement network effectively reassigns up to 1 reviewer from each paper in $P^{+}$ to a paper in either $P^{0}$ or $P^{-}$.
Additionally, up to 1 reviewer from each paper in $P^0$ may be reassigned to a paper in $P^{-}$.
As before, any edge in the optimal flow plan from a reviewer to a paper (or that paper's dummy node) corresponds to an assignment.
Any edge from a paper to a reviewer corresponds to \emph{unassigning} the reviewer from the corresponding paper.

Formally, we prove the following fact:
\begin{fact}
\label{fact:decrease}
After modifying an assignment according to the optimal flow plan in $\Gcal'$, no new papers will be added to $P^{-}$.
\end{fact}
\noindent The proof of Fact \ref{fact:decrease} appears in the appendix.

After solving MCF in the refinement network, some papers in
$P^{+}$ and $P^{-}$ may be assigned $\covcstr - 1$ reviewers, which violates the paper capacity constraints.
To make the assignment valid, solve MCF in the original flow network (Figure \ref{fig:rapflow}) with respect to the current assignment, the available reviewers, and the papers in violation.

\flowalg can only terminate after a valid solution has been constructed (i.e., after phase 1).
The three phases are repeated until either: a) there are no papers in $P^{-}$ or b) the number of papers in $P^{-}$ remains the same after two successive iterations. 

\paragraph{Load Lower Bounds.}
Incorporating reviewer load lower bounds can be done by adding a single step to \flowalg.  
Specifically, in phase 1, first construct a network where the capacity on the edge from $s$ to $r_i$ is $\loadlb_i$ (rather than $\loadub_i$).  
The total flow through the network is $\sum_{i=1}^{|R|}\loadlb_i$ and thus all load lower bounds are satisfied.
Once this initial flow plan is constructed, record the corresponding assignments and update the capacity of each edge between $s$ and $r_i$ to be $\loadub_i - \loadlb_i$.
Similarly, update the capacity of each edge between $p_j$ and $t$ to be the difference between the paper's coverage constraint and the number of reviewers assigned to $p_j$ in the initial flow plan.
The flow plan through the updated network, combined with the initial flow plan, constitute a valid assignment.
Afterwards, continue with phases 2 and 3 as normal.
The additional step must be performed in each invocation of phase 1.
\section{Experiments}
\label{sec:exp}

In this section we compare 4 paper matching algorithms:
\begin{enumerate}[noitemsep,topsep=0pt,parsep=0pt,partopsep=0pt,leftmargin=*]
    \item \textbf{\textsc{TPMS}} - optimal matching with respect to the \tpms.
    \item \textbf{\iralg} - our method, Algorithm \ref{alg:iralg}.
    \item \textbf{\flowalg} - our min-cost-flow-based algorithm (Section \ref{subsec:local-flow}).
    \item \textbf{\prfa~\cite{stelmakh2018peerreview4all}} - state-of-the-art flow-based paper matching algorithm that maximizes the minimum \papscore. For large problems we only run 1 iteration (\prfa (i1)).
\end{enumerate}
\textsc{TPMS}, \iralg and \prfa are implemented in Gurobi--an industrial mathematical programming toolkit~\cite{Gurobi-Optimization:2015aa}.
\flowalg is implemented using OR-Tools\footnote{ \url{https://developers.google.com/optimization/}}.

In our experiment we use data from 3 real conferences\footnote{Our data is anonymous and kindly provided by OpenReview.net and the Computer Vision Foundation.}.
Each dataset is comprised of: a matrix of paper-reviewer affinities (paper and reviewer identities are anonymous), a set of coverage constraints (one per paper), and a set of load upper bound constraints (one per paper).
One of our datasets also includes load lower bounds.
We do not evaluate \prfa on datasets when the load lower bounds are included since it was not designed for this scenario.

We report various statistics of each matching.
For completeness, we also include the runtime of each algorithm.
However, note that an algorithm's runtime is significantly affected by a number of factors, including: hardware, the extent to which the algorithm has been optimized, dataset on which it is run, etc.
All experiments are run on the same MacBook Pro with an Intel i7 processor.

\paragraph{Finding fairness thresholds.} Both \iralg and \flowalg take as input a fairness threshold, $T$.
Since the best value of this threshold is unknown in advance, we search for the best value using 10 iterations of binary search.
For \iralg, at iteration $i$ with threshold $T_i$, we use a linear programming solver to check whether there exists an optimal solution to the relaxation of the corresponding local fairness formulation.
By Fact \ref{fact:frac}, if a solution exists, then \iralg will successfully return an integer solution.
For \flowalg we do a similar binary search and return the threshold that led to the largest minimum \papscore.
In our implementation of \flowalg, when we test a new threshold $T$ during the binary search, we initialize from the previously computed matching.
Note that \prfa does not require such a threshold as an input parameter.

\paragraph{Matching profile boxplots.}
We visualize a matching via a set of \papscore quintiles, which we call it's \emph{profile}.
To construct the profile of a matching, compute the \papscore of each paper and sort in non-decreasing order.
The sorted list of scores is divided into 5 groups, each group containing an equal number of papers\footnote{Most datasets do not include a number of papers that is divisible by 5; in this case, the last quintile has fewer papers.}.
Each group of sorted \papscores is further divided into 4 even groups, $a, b, c$ and $d$ (with $a$ and $d$ containing the smallest and largest \papscores, respectively).
In each profile visualization that follows, the box in each column is defined by the minimum score in $b$, $b_{min}$, and maximum score in $c$, $c_{max}$ for the corresponding group (i.e, quintile).
The lowest horizontal line in a column is defined by the smallest \papscore that is greater than or equal to $b_{min} - \frac{c_{max} - b_{min}}{2}$; the highest horizontal line in the column is defined by the largest \papscore that is smaller than or equal to $c_{max} + \frac{c_{max} - b_{min}}{2}$.
The rest of the points are considered outliers and denoted by red \textsc{x}'s.
The median \papscore among $a,b,c$ and $d$ is represented as an orange line.
A matching's profile provides a visual summary of the distribution of \papscores it induces, including the best and worst \papscores.

\subsection{Medical Imaging and Deep Learning}
\label{sec:midl}
In our first experiment we use data from the Medical Imaging and Deep Learning (MIDL) Conference.
The data includes affinities of 177 reviewers for 118 papers.
The affinities range from -1.0 to 1.0.
Each paper must be reviewed by 3 reviewers and each reviewer must be assigned no more than 4 and no fewer than 2 papers (i.e., the data includes upper and lower bounds on reviewer loads).

\begin{figure*}[t]
\captionsetup[subfigure]{justification=centering}
\begin{subfigure}[h]{0.24\columnwidth}
  \centerline{\includegraphics[width=\columnwidth]{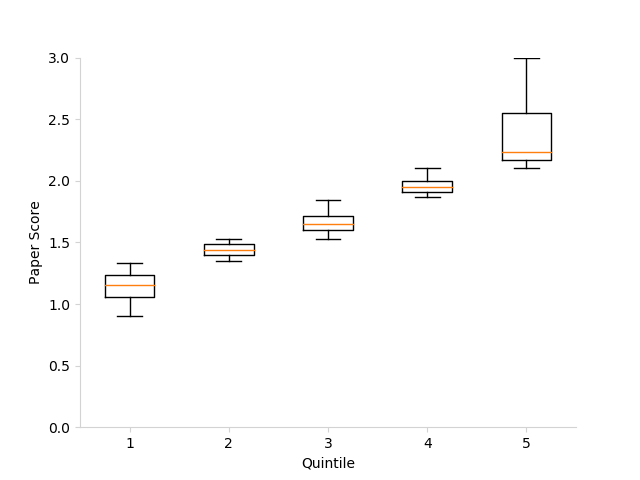}}
  \caption{\textsc{TPMS}.}
  \label{fig:midl-basic}
\end{subfigure}
\begin{subfigure}[h]{0.24\columnwidth}
  \centerline{\includegraphics[width=1.0\columnwidth]{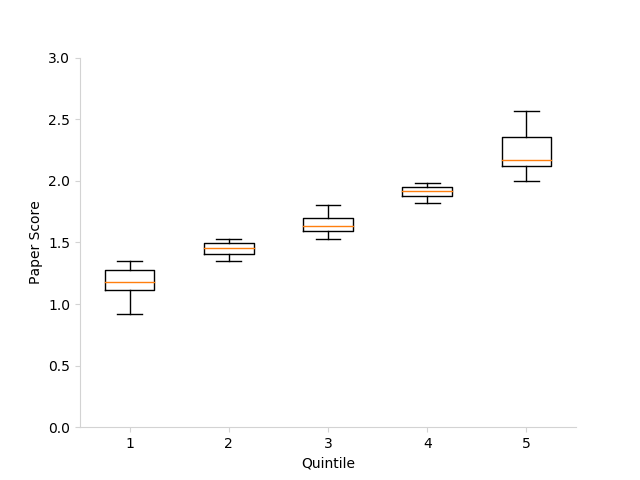}}
  \caption{\prfa.}
  \label{fig:midl-pr4a}
\end{subfigure}
\begin{subfigure}[h]{0.24\columnwidth}
  \centerline{\includegraphics[width=1.0\columnwidth]{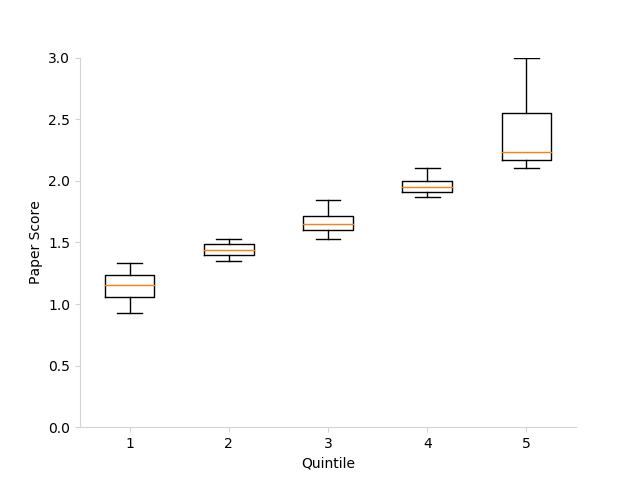}}
  \caption{\iralg.}
  \label{fig:midl-iralg}
\end{subfigure}
\begin{subfigure}[h]{0.24\columnwidth}
  \centerline{\includegraphics[width=1.0\columnwidth]{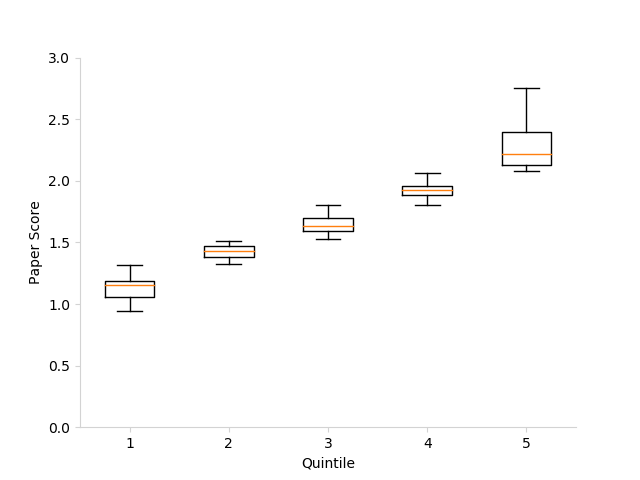}}
  \caption{\flowalg.}
    \label{fig:midl-msflow}
\end{subfigure}\\
\begin{subfigure}[h]{0.32\columnwidth}
  \centerline{\includegraphics[width=\columnwidth]{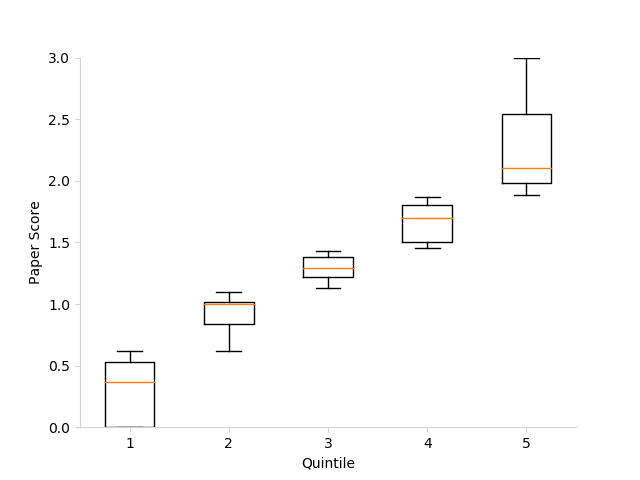}}
  \caption{\textsc{TPMS} with lower bounds.}
  \label{fig:midl-basic-lb}
\end{subfigure}
\begin{subfigure}[h]{0.32\columnwidth}
  \centerline{\includegraphics[width=1.0\columnwidth]{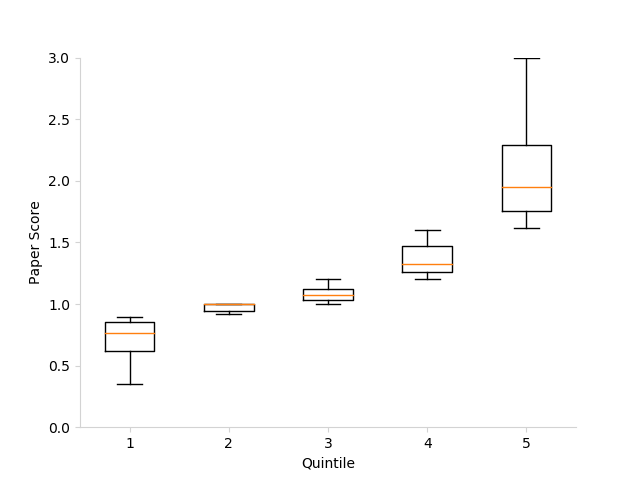}}
  \caption{\iralg with lower bounds.}
    \label{fig:midl-iralg-lb}
\end{subfigure}
\begin{subfigure}[h]{0.32\columnwidth}
  \centerline{\includegraphics[width=1.0\columnwidth]{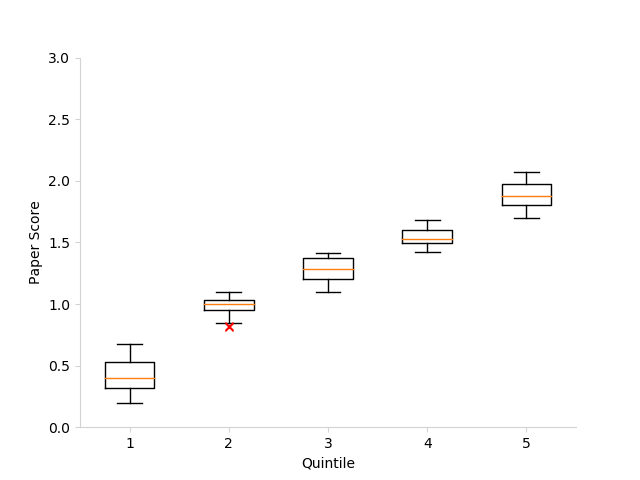}}
  \caption{\flowalg with lower bounds.}
  \label{fig:midl-msflow-lb}
\end{subfigure}
\caption{\textbf{\midl.} Figures \ref{fig:midl-basic}-\ref{fig:midl-msflow} visualize profiles of matchings computed by \textsc{TPMS}, \prfa, \iralg and \flowalg on the \midl data with reviewer load lower bounds excluded. All 4 algorithms construct similar profiles. Figures \ref{fig:midl-basic-lb}-\ref{fig:midl-msflow-lb} visualize profiles of matchings constructed with respect to the lower bounds. The introduction of lower bounds leads to many papers with low \papscores in the \textsc{TPMS} matching.}
\label{fig:midl}
\end{figure*}

\begin{table*}[tbh]
  \centering
  \scriptsize
\begin{tabular}{c c c c c c c c c c c c c}
\hline
Data & Bounds & Alg & Time (s) & Obj & Min PS & Max PS & Mean PS & Std PS & Min RA & Max RA & Std RA \\
\hline
        &Up          & \textsc{TPMS} & \emph{0.10}   & \emph{201.88}   & 0.90          & \emph{3.00}          & \textbf{1.71} & 0.45 & 0 & 4 & 1.80 \\
        &Up          & \prfa         & 293.83        & 197.32          & 0.92          & 2.57                 & 1.67          & \emph{\textbf{0.38}} & 0 & 4 & 1.79 \\
        &Up          & \iralg        & 1.60          & \textbf{201.83} & 0.93          & \emph{\textbf{3.00}} & \textbf{1.71} & 0.45 & 0 & 4 & 1.80 \\
\midl   & Up         & \flowalg      & \textbf{1.15} & 197.67          & \textbf{0.94} & 2.75                 & 1.68          & 0.41 & 0 & 4 & 1.80 \\
        &Lo + Up  & \textsc{TPMS} & \emph{0.17}   & \emph{150.04}   & 0.00          & \emph{3.00}          & \emph{1.27}   & 0.69 & 2 & 2 & 0.00 \\
        &Lo + Up  & \iralg        & 3.01          & \textbf{145.56} & \textbf{0.35} & \emph{\textbf{3.00}} & \textbf{1.23} & \textbf{0.50} & 2 & 2 & 0.00 \\
        &Lo + Up  & \flowalg      & \textbf{2.17} & 143.12          & 0.19          & 2.07                 & 1.21          & \textbf{0.50} & 2 & 2 & 0.00 \\
\hline
\hline
        &Up         & \textsc{TPMS}  & \emph{47.24}    & \emph{5443.64}   & 0.00                 & 3.00 & \emph{2.08}    & 1.07 & 0 & 6 & 0.82 \\
        &Up         & \prfa (i1)     & 3251.37         & 5134.08          & \textbf{0.77}        & 3.00 & 1.96           & \emph{\textbf{0.52}} & 0 & 6 & 1.24 \\
        &Up         & \iralg         & 594.51          & \textbf{5373.39} & 0.27                 & 3.00 & \textbf{2.05}  & 0.84 & 0 & 6 & 0.83 \\
\cvpra  & Up        & \flowalg       & \textbf{225.29} & 4444.95          & \textbf{0.77}        & 3.00 & 1.69           & 0.64 & 2 & 6 & \emph{\textbf{0.61}} \\
        &Lo + Up & \textsc{TPMS}  & \emph{49.62}    & \emph{5443.64}   & 0.00                 & 3.00 & \emph{2.08}    & 1.07 & 2 & 6 & 0.78 \\
        &Lo + Up & \iralg         & 694.03          & \textbf{5373.23} & 0.29                 & 3.00 & \textbf{2.05}  & 0.84 & 2 & 6 & 0.87 \\
        &Lo + Up & \flowalg       & \textbf{587.69} & 4339.60          & \emph{\textbf{0.94}} & 3.00 & 1.65           & \emph{\textbf{0.48}} & 3 & 6 & \emph{\textbf{0.63}} \\
\hline
\hline
        & Up         & \textsc{TPMS} & \emph{256.73}    & \emph{112552.11}   & 1.37                  & \textbf{29.24} & \emph{22.23}   & 5.52 & 0 & 9 & 2.97 \\
        &Up          & \prfa (i1)    &  8683.79         & 108714.98          & \emph{\textbf{12.68}} & 29.13          & 21.48          & \textbf{3.86} & 0 & 9 & 2.97 \\
        &Up          & \iralg        & 3785.64          & \textbf{112263.94} & 7.19                  & \textbf{29.24} & \textbf{22.18} & 4.75 & 0 & 9 & 2.96 \\  
2018  &Up          & \flowalg      & \textbf{910.08}  & 91029.66           & 11.12                 & 29.19          & 17.98          & 4.49 & 0 & 9 & \textbf{2.91} \\
        &Lo + Up  & \textsc{TPMS} & \emph{636.01}    & \emph{108634.18}   & 0.00                  & \textbf{29.24} & \emph{21.46}   & 6.28 & 2 & 9 & 1.66 \\
        &Lo + Up  & \iralg        & 4666.27    & \textbf{108083.00} & 7.17                  & \textbf{29.24} & \textbf{21.35} & 5.06 & 2 & 9 & 1.67 \\ 
        &Lo + Up  & \flowalg      & \textbf{1790.71} & 86166.07           & \emph{\textbf{10.52}} & 22.79          & 17.02          & \textbf{2.77} & 2 & 9 & \textbf{1.61} \\
\hline
\end{tabular}
\caption{\textbf{\midl, \cvpra and \cvprb (2018) matching statistics.} PS denotes a \papscore and RA denotes the number of papers assigned to a reviewer. Values in \emph{italic} represent the best performer; \textbf{bold} values indicate the best of the fair matching algorithms. \flowalg is always the fastest of the fair algorithms while \iralg always achieves the highest objective score.  \flowalg is always competitive with or outperforms \prfa in terms of the fairness with respect to the most disadvantaged paper.}
\label{fig:big-tbl}
\end{table*}

Figure \ref{fig:midl} displays the profiles of matchings computed by the 4 algorithms with and without lower bounds.
Without lower bounds, all algorithms produce similar profiles, except that the maximum \papscore achieved by \prfa and \flowalg are lowest.
Somewhat similarly, these two algorithms achieve lower objective scores, which is likely a result of the fact that neither explicitly maximizes the global sum of \papscores.
Interestingly, \textsc{TPMS} constructs a matching that is relatively fair with respect to \papscores even though it is not designed to do so.

When lower bounds are considered, the algorithms produce much different profiles.
First, notice that \textsc{TPMS} constructs a matching in which some papers have a corresponding \papscore of 0--signaling an unfair assignment.
Of the fair matching algorithms, \iralg's profile includes a higher minimum \papscore, a higher maximum \papscore, and a higher objective score.
However, \iralg is ~40\% slower than \flowalg.
Also note that on this small dataset, we run \prfa with no upper bound on the number of iterations (hence the long runtime).
Table \ref{fig:big-tbl} (first block) contains matching statistics of the various algorithms for \midl.

\subsection{\cvpra}
\label{sec:cvpr}
Our next experiment is performed with respect to data from a previous year's Conference on Computer Vision and Pattern Recognition (\cvpra).
The data includes the affinities of 1373 reviewers for 2623 papers, which amounts to a substantially larger problem than that posed by the \midl data.
All affinities are between 0.0 and 1.0.
As before, each paper must be reviewed by 3 different reviewers.
Each reviewer may not be assigned to more than 6 papers.
Our data does not contain lower bounds.
For the purpose of demonstration, we construct a set of synthetic reviewer load lower bounds where all reviewers must review at least 2 papers.

\begin{figure*}[t]
\captionsetup[subfigure]{justification=centering}
\begin{subfigure}[h]{0.24\columnwidth}
  \centerline{\includegraphics[width=\columnwidth]{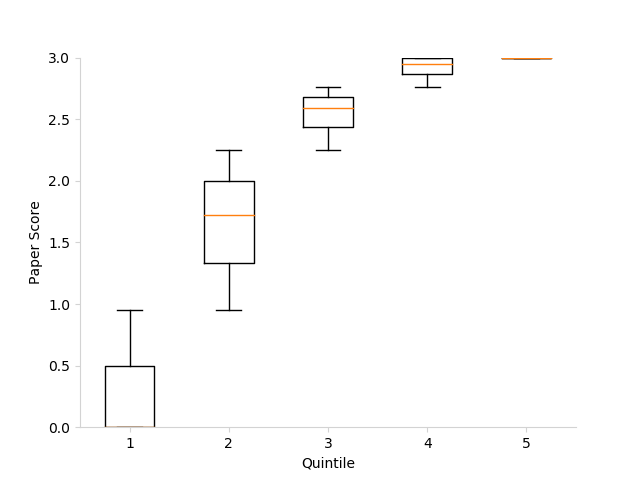}}
  \caption{\textsc{TPMS}.}
  \label{fig:cvpr-basic}
\end{subfigure}
\begin{subfigure}[h]{0.24\columnwidth}
  \centerline{\includegraphics[width=1.0\columnwidth]{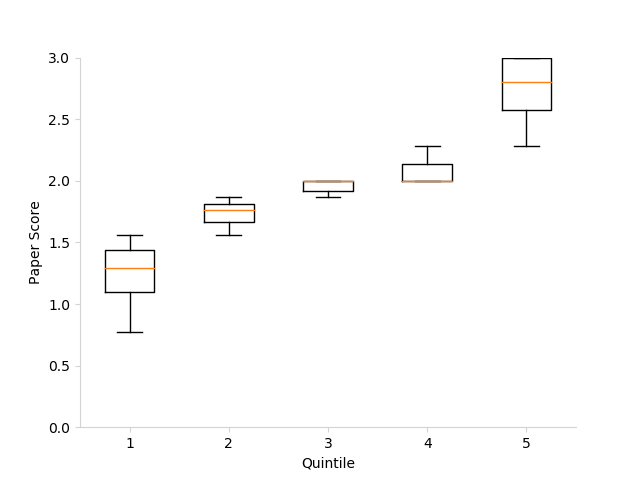}}
  \caption{\prfa (1 Iteration).}
  \label{fig:cvpr-pr4a}
\end{subfigure}
\begin{subfigure}[h]{0.24\columnwidth}
  \centerline{\includegraphics[width=1.0\columnwidth]{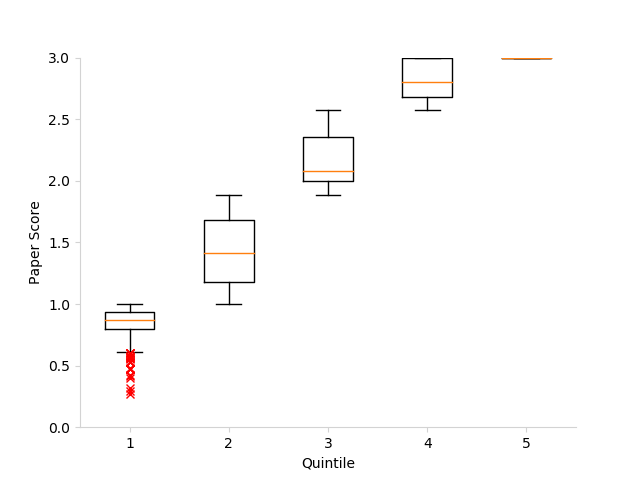}}
  \caption{\iralg.}
  \label{fig:cvpr-iralg}
\end{subfigure}
\begin{subfigure}[h]{0.24\columnwidth}
  \centerline{\includegraphics[width=1.0\columnwidth]{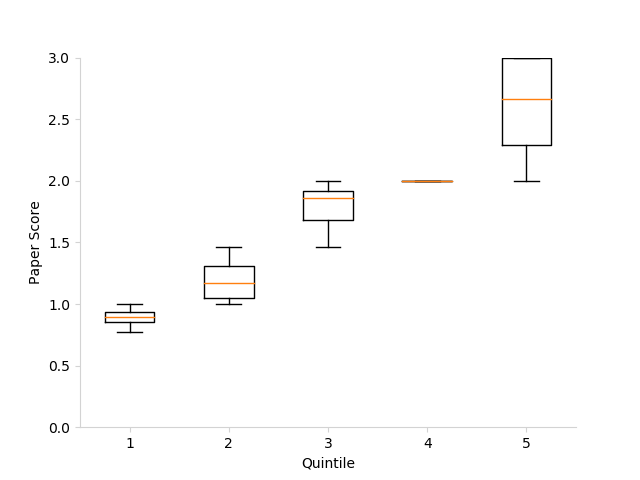}}
  \caption{\flowalg.}
    \label{fig:cvpr-msflow}
\end{subfigure}\\
\begin{subfigure}[h]{0.32\columnwidth}
  \centerline{\includegraphics[width=\columnwidth]{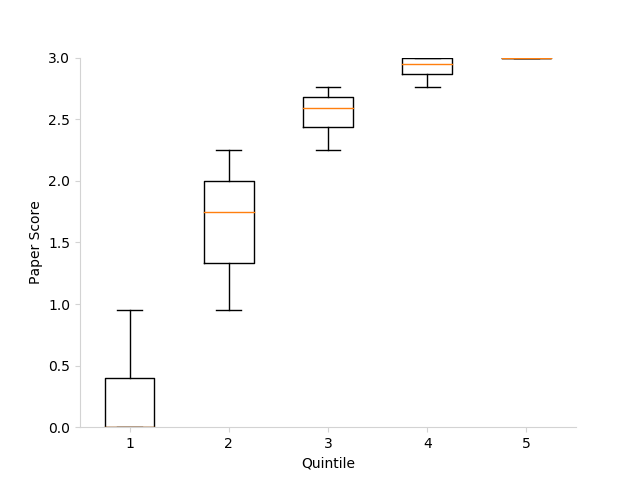}}
  \caption{\textsc{TPMS} with lower bounds.}
  \label{fig:cvpr-basic-lb}
\end{subfigure}
\begin{subfigure}[h]{0.32\columnwidth}
  \centerline{\includegraphics[width=1.0\columnwidth]{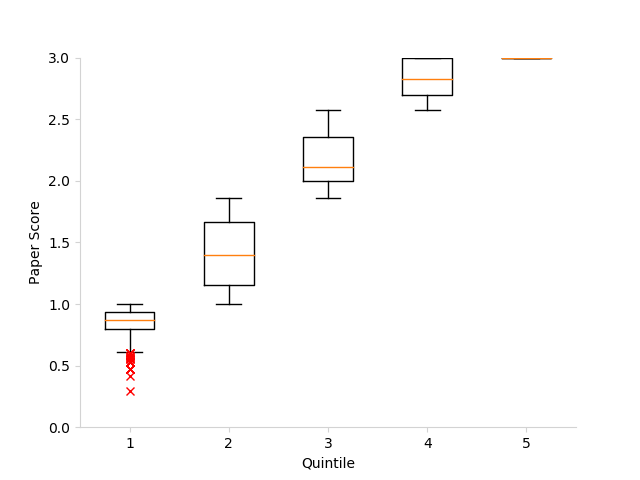}}
  \caption{\iralg with lower bounds.}
    \label{fig:cvpr-iralg-lb}
\end{subfigure}
\begin{subfigure}[h]{0.32\columnwidth}
  \centerline{\includegraphics[width=1.0\columnwidth]{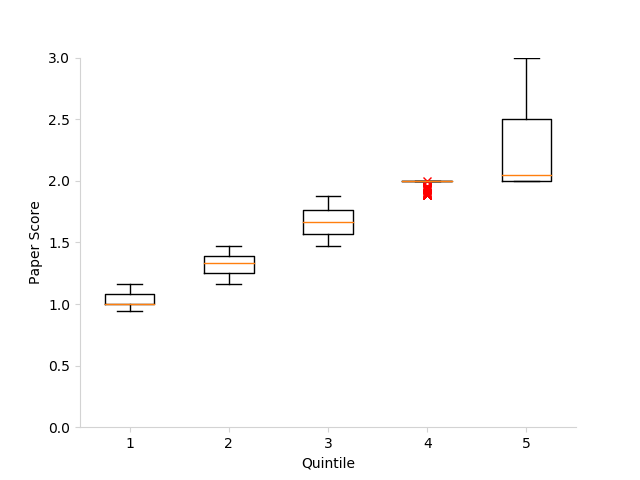}}
  \caption{\flowalg with lower bounds.}
  \label{fig:cvpr-msflow-lb}
\end{subfigure}
\caption{\textbf{Matching Profiles for \cvpra.}}
\label{fig:cvpr}
\end{figure*}

The results are contained in Figure \ref{fig:cvpr} and Table \ref{fig:big-tbl} (second block).
As before, \flowalg is the fastest fair matching algorithm, achieving 2x speedup over \iralg and an order of magnitude speedup over \prfa when lower bounds are excluded.
When lower bounds are included, \flowalg is still ~100s (15\%) faster than \iralg.
\prfa and \iralg achieve similar fairness.
Interestingly, \flowalg finds the matching with highest degree of fairness when lower bounds on reviewing loads are applied.
However, this comes at the expense of a relatively low objective score.
\iralg constructs a more fair matching than \textsc{TPMS}, but not than the other two fair matching algorithms.
This is unsurprising because \iralg optimizes the global objective, unlike the other algorithms, which more directly optimize fairness.
\iralg's balance between fairness and global optimality is illustrated by \iralg's profile (Figure \ref{fig:cvpr-iralg-lb}), which contains a handful outliers with low scores, but many papers with comparatively high \papscore in quintiles 3, 4 and 5.

\subsection{\cvprb}
\label{sec:large}
In our final experiment, we use data from CVPR 2018 (\cvprb).
The data contains the affinities of 2840 reviewers for 5062 papers--a substantial increase in problem size over \cvpra.
Affinities range between 0.0 and 11.1, with many scores closer to 0.0 (the mean score is ~0.36).
Each paper must be reviewed 3 times.
Reviewer load upper bounds vary by reviewer and range between 2.0 and 9.0.
Again, the data does not include load lower bounds and so we construct synthetic lower bounds of 2.0 for all reviewers.
Because of the size of the problem, the binary search for a suitable value of $T$ did not terminate within 5 hours.
Therefore, we select $T$ by summing the minimum \papscore found by \flowalg and $\frac{1}{2}A_{max}$.
The reported run time includes the run time of \flowalg.

Table \ref{fig:big-tbl} (third vertical block) reveals similar trends with respect to speed (\flowalg is most efficient) and fairness (\prfa and \iralg are the most fair).
Figure \ref{fig:cvpr2018} displays the corresponding matching profiles.
\begin{figure*}[t]
\captionsetup[subfigure]{justification=centering}
\begin{subfigure}[h]{0.24\columnwidth}
  \centerline{\includegraphics[width=\columnwidth]{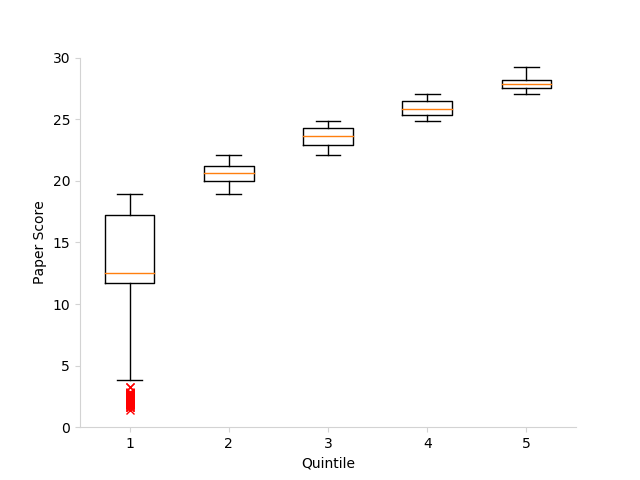}}
  \caption{\textsc{TPMS}.}
  \label{fig:cvpr2018-basic}
\end{subfigure}
\begin{subfigure}[h]{0.24\columnwidth}
  \centerline{\includegraphics[width=1.0\columnwidth]{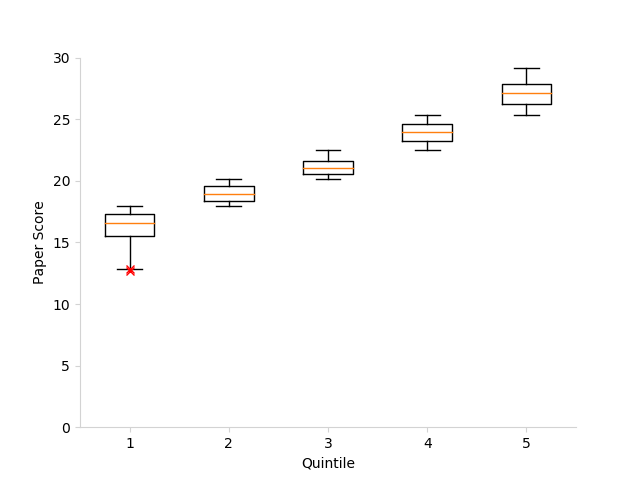}}
  \caption{\prfa (1 Iteration).}
  \label{fig:cvpr2018-pr4a}
\end{subfigure}
\begin{subfigure}[h]{0.24\columnwidth}
  \centerline{\includegraphics[width=1.0\columnwidth]{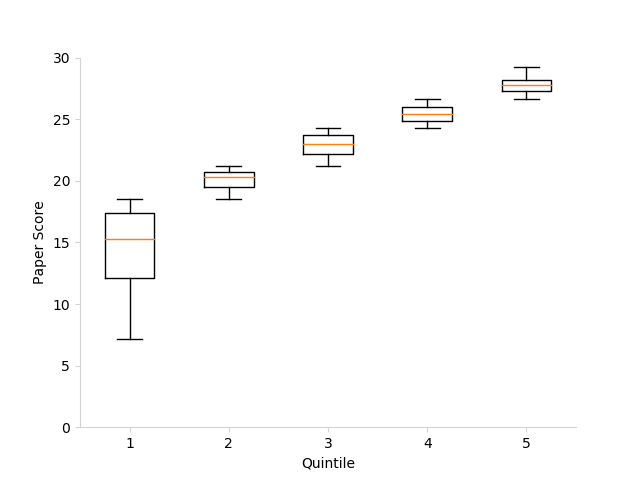}}
  \caption{\iralg.}
  \label{fig:cvpr2018-iralg}
\end{subfigure}
\begin{subfigure}[h]{0.24\columnwidth}
  \centerline{\includegraphics[width=1.0\columnwidth]{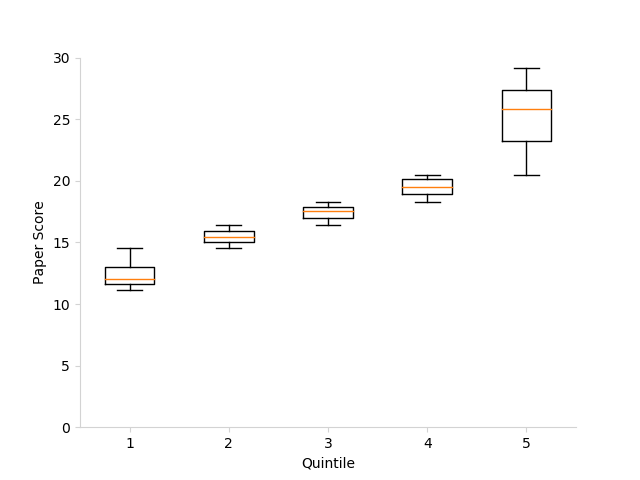}}
  \caption{\flowalg.}
    \label{fig:cvpr2018-msflow}
\end{subfigure}\\
\begin{subfigure}[h]{0.32\columnwidth}
  \centerline{\includegraphics[width=\columnwidth]{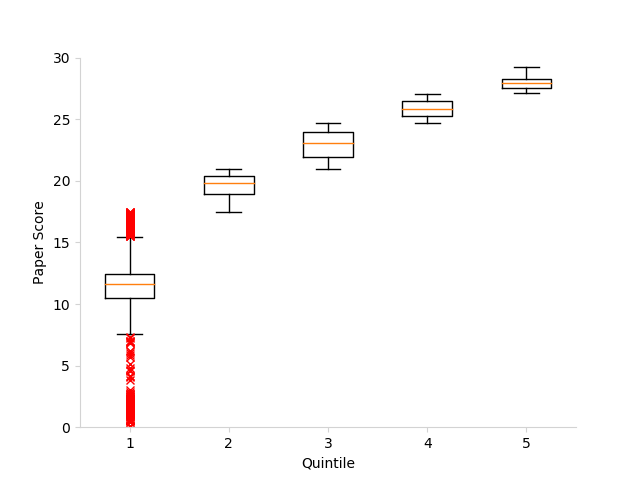}}
  \caption{\textsc{TPMS} with lower bounds.}
  \label{fig:cvpr2018-basic-lb}
\end{subfigure}
\begin{subfigure}[h]{0.32\columnwidth}
  \centerline{\includegraphics[width=1.0\columnwidth]{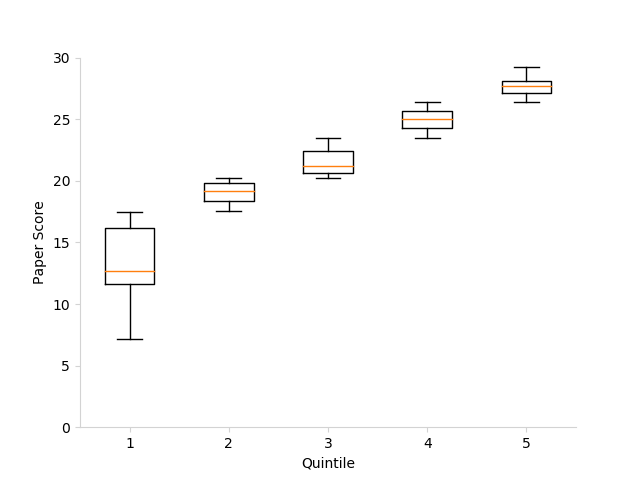}}
  \caption{\iralg with lower bounds.}
    \label{fig:cvpr2018-iralg-lb}
\end{subfigure}
\begin{subfigure}[h]{0.32\columnwidth}
  \centerline{\includegraphics[width=1.0\columnwidth]{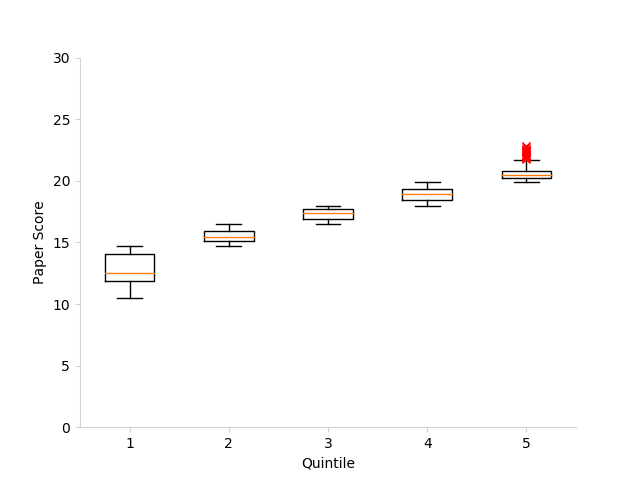}}
  \caption{\flowalg with lower bounds.}
  \label{fig:cvpr2018-msflow-lb}
\end{subfigure}
\caption{\textbf{Matching Profiles for \cvprb.}}
\label{fig:cvpr2018}
\end{figure*}
\section{Related Work}
\label{sec:related}
Our work is most similar to previous studies that develop algorithms for constructing fair assignments for the RAP.
Two studies propose to optimize for fairness with respect to the least satisfied reviewer, which can be formulated as a maximization over the minimum \papscore with respect to an assignment~\cite{garg2010assigning, stelmakh2018peerreview4all}.
The first algorithm, to which we compare, is \prfa~\cite{stelmakh2018peerreview4all}.
\prfa iteratively solves maximum-flow through a sequence of specially constructed networks, like our \flowalg, and is guaranteed to return a solution that is within a bounded multiplicative constant of the optimal solution with respect to their maximin objective.
As demonstrated in experiments, \flowalg is faster than \prfa and achieves similar quality solutions on data from real conferences.
We note that the work introducing \prfa also presents a statistical study of the acceptance of the \emph{best} papers among a batch submitted; we do not focus on paper acceptance in this work.

The second work proposes a rounding algorithm and prove an additive, constant factor approximation of the optimal assignment, like we do~\cite{garg2010assigning}.  
We note that both their algorithm and proof techniques are different from ours.
However, their algorithm requires solving a new linear program for each reviewer during each iteration, which is unlikely to scale to large problems.
Moreover, \prfa directly compares favorably to this algorithm~\cite{stelmakh2018peerreview4all}.

With respect to fairness, the creators of TPMS perform experiments that enforce load equity among reviewers (i.e., each reviewer should be assigned a
similar number of papers) via adding penalty terms to the objective~\cite{charlin2012framework}. 
These researcher, and others, explore formulations that maximize the minimum affinity among all assigned reviewers, which is different from our fairness constraint~\cite{o2005paper, wang2008survey}.
Others have posed instances of the RAP that require at least one reviewer assigned to each paper to have an affinity greater than $T$.
In this setting, one classic piece gives an algorithm for constructing assignments that maximizes $T$ by modeling the RAP as a transshipment problem~\cite{hartvigsen1999conference}.
Other objectives have been considered for the RAP, but these tend to be global optimizations with no local constraints that can lead to certain papers being assigned groups of inappropriate reviewers~\cite{goldsmith2007ai, wang2008survey, lian2018conference}.

Some previous work on the RAP models each paper as a
binary set of topics and each reviewer as a binary set of expertises (the overall sets of topics and expertises are the same). 
In this setting the goal to maximize coverage of each paper's topics by the assigned reviewers' expertises~\cite{merelo2004conference, karimzadehgan2009constrained,
  long2013good}. 
A generalized settings allows paper and reviewer representations to be real-valued vectors rather than binary~\cite{tang2010expertise, kou2015topic}. 
The resulting optimization problems are solved via ILPs, constraint based optimization or greedy algorithms.
While representing papers and reviewers as topic vectors allows for more fine-grained characterization of affinity, in practice, reviewer-paper affinity is typically represented by a single real-value--like the real-conference data we use in experiments.

A significant portion of the work related to the RAP explores methods for modeling reviewer-paper affinities. 
Some of the earliest work employs latent semantic indexing with respect to the abstracts of submitted and previously published papers~\cite{dumais1992automating}. 
More recent work models each author as a mixture of personas and each persona as a mixture of topics; each
paper written by an author is generated from a combination of
personas~\cite{mimno2007expertise}.  
Other approaches use reviewer bids to derive the affinity between papers and reviewers. 
Since reviewers normally do not bid on all papers, collaborative filtering has been used for bid imputation~\cite{conry2009recommender}. 
Finally, some approaches model affinity using proximity in coauthorship networks, citations counts, and the venues in which a paper is published~\cite{rodriguez2008algorithm,li2013automatic, liu2014robust}.
\section{Conclusion}
This work introduces the local fairness formulation of the reviewer assignment problem (RAP) that includes a global objective as well as local fairness constraints.
Since it is NP-Hard, we present two algorithms for solving this formulation.
The first algorithm, \iralg, relaxes the formulation and employs a specific rounding technique to construct a valid matching.
Theoretically, we show that \iralg violates fairness constraints by no more than the maximum reviewer-paper affinity, and may only violate load constraints by 1.
The second algorithm, \flowalg, is a more efficient heuristic that operates by solving a sequence of min-cost flow problems.
We compare our two algorithms to standard matching techniques that do not consider fairness, and a state-of-the-art algorithm that directly optimizes for fairness.
On 3 datasets from recent conferences, we show that \iralg is best at jointly optimizing the global matching while statisfying fairness constraints, and \flowalg is the most efficient of the fairness matching algorithms.
Despite a lack of theoretical guarantees, \flowalg constructs highly fair matchings.

All code for experiments is available here: \url{https://github.com/iesl/fair-matching}.\newline
Anonymized data is either included in the repository or available upon request from the first author.

\section{Acknowledgments}
This material is based upon work supported in part by the Center for Data Science and the Center for Intelligent Information Retrieval, and in part by the Chan Zuckerberg Initiative under the project "Scientific Knowledge Base Construction." B. Saha was supported in part by an NSF CAREER award (no. 1652303), in part by an NSF CRII award (no. 1464310), in part by an Alfred P. Sloan Fellowship, and in part by a Google Faculty Award.  Opinions, findings and conclusions/recommendations expressed in this material are those of the authors and do not necessarily reflect those of the sponsors.

\newpage

\bibliographystyle{ACM-Reference-Format}
\bibliography{ms}

\newpage

\begin{appendix}
\section{\iralg Guarantees}
We restate and then prove Theorem \ref{thm:imma}.

\begin{theorem*}
  Given a feasible instance of the local fairness formulation
  $\Pcal = <R, P, \loadlb, \loadub, \covcstr,A,T>$, \iralg returns
  an integer solution in which each local fairness constraint may be
  violated by at most $A_{max}$, each load constraint may be
  violated by at most 1 and the global objective is maximized.
\end{theorem*}
The local fairness formulation,  $\Pcal$, is comprised of a set of
reviewers, $R$, a set of papers, $P$, reviewer load lower and upper bounds, $\loadlb$ and $\loadub$, respectively, coverage constraints, $\covcstr$, a paper-reviewer affinity matrix, $A$, and a local fairness threshold, $T$.
To prove this theorem we rely on three lemmas.
The first guarantees that \iralg does not violate a load constraint by more than 1; the second guarantees that \iralg will never violate a local fairness constraint by more than $A_{max}$; the third guarantees that \iralg will always terminate if the input problem is feasible.

\begin{lemma}
  Given a feasible instance of the local fairness formulation, \iralg
  never violates a load constraint by more than 1.
  \label{thm:load}
\end{lemma}
\begin{proof}
  \label{proof:load}
  \iralg only drops load constraints if a reviewer is assigned fractionally to at most 2 papers. Clearly, if a
  reviewer is assigned to exactly one paper, the load constraint can be violated by at most one. Therefore,
  let $r_i$ be a reviewer, assigned fractionally to $p_j$ and $p_k$ only.
  Then,
  \begin{align*}
    \loadlb_i \le x_{ij} + x_{ik} + \alpha \le \loadub_i.
  \end{align*}
  where $\alpha$ is the total load on $r_i$ excluding $x_{ij}$ and $x_{ik}$.
  Since $r_i$ is only fractionally assigned to 2 papers, $\alpha$ must be integer; since $x_{ij}, x_{ik} \in (0, 1)$, $x_{ij} + x_{ik} < 2$.
  Thus,
  \begin{align*}
    \loadlb_i - 1 \le \alpha \le \loadub_i - 1.
  \end{align*}
  If the load constraints are dropped and $r_i$ is neither assigned to $p_j$ nor $p_k$, then $r_i$ will retain a load of $\alpha$, which is at least as large as 1 less than $\loadlb_i$.
  On the other hand, if $r_i$ is assigned to both $p_j$ and $p_k$, then $r_i$ will exhibit a load of $\alpha + 2 \le \loadub_i + 1$.
\end{proof}

\begin{lemma}
  Given a feasible instance of the local fairness formulation, \iralg
  never violates a local fairness constraint by more than $A_{max}$.
  \label{thm:violation}
\end{lemma}
\begin{proof}
\label{proof:violation}
\iralg only drops a paper's local fairness constraint if that paper has
at most 3 reviewers fractionally assigned to it. Clearly, if a paper has only one reviewer fractionally assigned to it, the local
fairness constraint can be violated by at most $A_{max}$.
Assume during an iteration of \iralg a paper has exactly 2 reviewers fractionally assigned to it.
Call that paper $p_k$ and those reviewers $r_i$ and $r_j$.
During each iteration of \iralg, a feasible solution to the relaxed local fairness formulation is computed.
Therefore,
\begin{align*}
C' + x_{ik} + x_{jk} = \covcstr_k,
\end{align*}
where $C'$ is load the on $p_k$ aside from the load contributed by reviewers $r_i$ and $r_j$.
Recall that $x_{ik}, x_{jk} \in (0, 1)$ and $r_i$ and $r_j$ are the only reviewers fractionally assigned to $p_k$.
Therefore $x_{ik} + x_{jk} = 1$.
Moreover,
\begin{align*}
    x_{ik}A_{ik} + x_{jk}A_{jk} \le x_{ik}A_{max} + x_{jk}A_{max} = A_{max}.
\end{align*}
Now, consider the \papscore at $p_k$, and let $T'$ be the total affinity between $p_k$ and all its assigned reviewers, except for $r_i$ and $r_j$.
Then,
\begin{align*}
  T' + x_{ik}A_{ik} + x_{jk}A_{jk} &\ge T\\
  T' &\ge T - x_{ik}A_{ik} - x_{jk}A_{jk}\\
  &\ge T - A_{max}.
\end{align*}
Since either $r_i$ or $r_j$ must be assigned integrally to $p_k$ (lest the coverage constraint be violated), dropping the local fairness constraint at $p_k$ can only lead to a violation of the local fairness constraint at $p_k$ by at most $A_{max}$.

Next, consider the case that $p_k$ has 3 reviewers fractionally assigned to it, $r_h$, $r_i$ and $r_j$.
Since the coverage constraint at $p_k$ must be met with equality, one of the two cases below must be true:
\begin{align*}
    x_{hk} + x_{ik} + x_{jk} = 1
\end{align*}
or
\begin{align*}
    x_{hk} + x_{ik} + x_{jk} = 2.
\end{align*}

As before, let $T'$ be the \papscore at $p_k$, excluding affinity contributed from fractionally assigned reviewers.
If the first case above is true, then $x_{hk}A_{hk} + x_{ik}A_{ik} + x_{jk}A_{jk} \le A_{max}$.
Furthermore,
\begin{align*}
  T' + x_{hk}A_{hk} + x_{ik}A_{ik} + x_{jk}A_{jk} &\ge T\\
  T' &\ge T - x_{hk}A_{hk} - x_{ik}A_{ik} - x_{jk}A_{jk}\\
  &\ge T - A_{max}.
\end{align*}
This means that even if all three reviewers were unassigned from $p_k$ (which would make satisfying the coverage constraint at $p_k$ impossible), the local fairness constraint would only be violated by at most $A_{max}$.
Now, consider case 2 above, where $x_{hk}A_{hk} + x_{ik}A_{ik} + x_{jk}A_{jk} \le 2A_{max}$.
In order to satisfy the coverage constraint at $p_k$, at least two of the three reviewers must be assigned integrally to $p_k$.
Without loss of generality, assume that
\begin{align*}
    A_{hk} = \max[A_{hk}, A_{ik}, A_{jk}] \le A_{max}.
\end{align*}
Even if $r_h$ is unassigned from $p_k$, the change in \papscore at $p_k$ is at most $A_{max}$ and the local fairness can be violated at most by $A_{max}$.
The same is also true if either $r_i$ or $r_j$ is unassigned from $p_k$.
\end{proof}

\begin{lemma}
  Given a feasible instance of the local fairness fromulation, \iralg
  always terminates.
  \label{thm:termination}
\end{lemma}
The goal in proving Lemma~\ref{thm:termination} is to show that during each iteration of \iralg, either: a constraint is dropped or an integral solution is found.
Before proving Lemma ~\ref{thm:termination} recall that the solution,
$x^\star$, of a linear program is always a \emph{basic feasible solution},
i.e., it has $n$ linearly independent tight constraints. Formally,
\begin{corollary}
  \label{cor:active-constraints}
  If $x^\star$ is a basic feasible solution of linear program $\Pcal$, then
  the number of non-zero variables in $x^\star$ cannot be greater than the
  number of linearly independent active constraints in $\Pcal$.
\end{corollary}

\begin{proof}
  \label{proof:termination}
  According to Algorithm\ref{alg:iralg}, \iralg drops constraints during any iteration in which it constructs a solution exhibiting at least one paper with at most 3 reviewers fractionally assigned to it or at least one reviewer assigned fractionally to at most 2 papers. If \iralg is able to drop a constraint or round a new variable to integral, it makes progress.
Therefore, \iralg could only fail to make progress if each reviewer was assigned fractionally to at least 3 papers and each paper was assigned fractionally to at least 4 reviewers.
  In the following, we show that this is impossible, using a particular invocation of Corollary~\ref{cor:active-constraints}.

  Assume for now that each reviewer is fractionally assigned to exactly 3 papers and each paper is assigned fractionally to exactly 4 reviewers.
  Therefore, the total number of fractional assignments can be written as follows:
  \begin{align*}
      \frac{1}{2}[3|R| + 4|P|].
  \end{align*}
  An instance of the local fairness paper matching problem contains an upper and lower bound constraint for each reviewer, 1 coverage constraint for each paper, and 1 local fairness constraint for each paper yielding $2|R| + 2|P|$
  total constraints.
  Note that for a reviewer $r$, only one of its load constraints (i.e., upper or lower) may be tight--assuming that the upper and lower bounds are distinct.
  Thus, an upper bound on the number of \emph{active} constraints is $|R| + 2|P|$.
  However, this means that the number of fractional variables is larger than the number of constraints:
    \begin{align*}
      \frac{1}{2}\left[3|R| + 4|P|\right] = \frac{3}{2}|R| + 2|P| > |R| + 2|P|
  \end{align*}
  which violates Corollary~\ref{cor:active-constraints}. When reviewers may be fractionally assigned to at least 3 papers and each paper is assigned fractionally to at least 4 reviewers, the number of nonzero fractional variables could only be larger. Note that, when there is no local fairness constraint \iralg returns an integral solution since the underlying constraint matrix becomes totally unimodular.

\end{proof}

Now to end the proof of the theorem, we note that the global objective value never decreases in subsequent rounds, as we always relax the formulation by dropping constraints and fix those integrality constraints for which $x_{i,j}$s have been returned as integer. Thus, \iralg maximizes the global objective.

\section{Proof of Fact~\ref{fact:decrease}}
\begin{proof} By definition, papers that are members of $P^{+}$ have paper score greater than $T$.
Therefore, unassigning a reviewer from a paper in $P^+$ may reduce the corresponding paper score by at most $A_{max}$ yielding a paper score of at least $T - A_{max}$, which makes the paper either a member of $P^{0}$ or $P^+$.
Now, consider the papers in $P^0$.
By step 7 above, a reviewer $r$ can only be unassigned from a paper $p \in P^0$ if the flow entering $p$ from $p'$ is large enough to make $p$'s resulting paper score at least as large as $T-A_{max}$.
Thus, the papers in $P^{0}$ either remain in $P^0$ or become members of $P^{+}$, which completes the proof.
\end{proof}

\end{appendix}

\end{document}